\documentclass[sigconf, authorversion]{acmart}
\usepackage{myCommands}

\AtBeginDocument{%
  }

\copyrightyear{2026}
\acmYear{2026}
\setcopyright{cc}
\setcctype{by}
\acmConference[ISSAC '26]{51st International Symposium on Symbolic and Algebraic Computation}{July 13--17, 2026}{Oldenburg, Germany}
\acmBooktitle{51st International Symposium on Symbolic and Algebraic Computation (ISSAC '26), July 13--17, 2026, Oldenburg, Germany}
\acmDOI{10.1145/3815436.3815460}
\acmISBN{979-8-4007-2595-1/2026/07}

\begin{document}

\title{On Minimum CADs for Algebraic Sets in Dimension Three}


\author{Lucas Michel}
\orcid{0000-0002-4115-7296}
\affiliation{%
  \institution{University of Liège}%\\  Mathematics Research Unit}
  %\city{Hekla}
  \country{Belgium}
  }
\email{lucas.michel@uliege.be}

%\renewcommand{\shortauthors}{L. Michel}

\begin{abstract}
    Cylindrical Algebraic Decomposition (CAD) algorithms typically produce a decomposition adapted to a finite family of semi-algebraic sets $\mathcal{F}$ (i.e. every member of $\mathcal{F}$ is a union of cells). Different algorithms may produce different outputs, and introduce unnecessary cell divisions. Recent work by Michel, Mathonet, and Zénaïdi in ISSAC 2024 formalised this issue by studying the refinement order on the set of all CADs adapted to $\mathcal{F}$ and analysing the existence of a minimum (coarsest) adapted CAD. It was shown that such a minimum adapted CAD always exists for subsets of $\mathbb{R}$ and $\mathbb{R}^2$, but not of $\mathbb{R}^n$ ($n \geqslant 3$) in general.
    
    It is natural to seek natural classes of subsets of $\mathbb{R}^n$ that admit a minimum adapted CAD.     
    In this paper, we identify a class of subsets of $\mathbb{R}^3$ that contains all algebraic sets for which minimum adapted CADs do exist. This provides the first positive existence theorem for minimum CAD for a non-trivial class of sets. 
\end{abstract}

%    <concept>
%        <concept_id>10010147.10010148.10010164</concept_id>
%        <concept_desc>Computing methodologies~Representation of mathematical objects</concept_desc>
%        <concept_significance>500</concept_significance>
%        </concept>
%  </ccs2012>

% \ccsdesc[500]{Computing methodologies~Representation of mathematical objects}

\keywords{Cylindrical Algebraic Decomposition, Algebraic set, Minimum CAD}



\maketitle

\section{Introduction}
Cylindrical Algebraic Decomposition (CAD) (introduced by Collins \cite{collins1975} for quantifier elimination over real closed fields), is a cornerstone of computational real algebraic geometry and plays a central role in several areas of computer algebra, including quantifier elimination, symbolic computation, and SMT solving. Its importance is reflected in the large body of work devoted to improving CAD-based methods, both theoretically and practically. As a consequence, different CAD algorithms may produce different decompositions for the same input, often with superfluous cell divisions. In this paper, we are interested in CADs that are as coarse as possible, and hence free of such superfluous cell divisions.

A CAD can be seen as a datastructure to represent semi-algebraic sets. Given a finite family $\Fc$ of semi-algebraic subsets of $\R^n$, a CAD $\Cr$ is adapted to (or represents) $\Fc$ if every member of $\Fc$ is a union of cells of $\Cr$. In \cite{miniCAD} and the extended version \cite{miniCAD-extended}, we introduced a formal framework for comparing CADs adapted to $\Fc$ through the refinement ordering (a CAD $\Cr$ is smaller than a CAD $\Dr$ if $\Cr$ is coarser than $\Dr$). This allows one to reason about CADs from an order-theoretic perspective and to distinguish between minimal CADs, which cannot be coarsened while remaining adapted, and minimum CADs, which are canonical, redundancy-free CAD that are coarser than any other adapted CAD. It was proved that every (finite family of) semi-algebraic set in $\mathbb{R}$ and $\mathbb{R}^2$ admit a minimum CAD. However, this property fails in higher dimensions: explicit counterexamples in $\mathbb{R}^n$ ($n \geqslant 3$) exhibit several semi-algebraic sets that admits several distinct minimal CADs, and hence no minimum one. Thus, in dimensions three and higher, semi-algebraic sets need not possess such a canonical CAD.

Since CAD algorithms typically take as input a finite family of multivariate polynomials and produce a decomposition adapted to all algebraic varieties they define, it is natural to ask whether such finite families always admit a minimum CAD. In other words, when restricting attention to algebraic (rather than arbitrary semi-algebraic) sets, can we expect the existence of a canonical optimal decomposition?

In this paper, we give a positive answer in dimension three and obtain a stronger statement. We prove in Theorem \ref{thrm:main-curtained} that every finite family of closed and curtained (defined below) semi-algebraic sets in $\mathbb{R}^3$ admits a minimum CAD. This extends the existence results of $\mathbb{R}$ and $\mathbb{R}^2$ and shows that the non-existence phenomena identified in $\mathbb{R}^3$ cannot occur under these  topological/algebraic/geometric constraints. As a consequence, we obtain the following main result of the paper. 

\begin{theorem}\label{thrm:main}
    Every finite family of algebraic sets in the space $\R^3$ admits a minimum CAD.
\end{theorem}

The only criterion currently available for showing that a given finite family $\Fc$ admits a minimum CAD is Theorem~6.3 of \cite{miniCAD-extended}. However, its application depends heavily on $\Fc$, which prevents its use in our setting.
In this paper, we instead provide a proof of Theorem~\ref{thrm:main} that is independent of Theorem~6.3 of \cite{miniCAD-extended}, relying on some confluence property and topological arguments.
%This result provides the first positive existence theorem for minimum CADs in dimension three, thereby giving a partial answer to questions left open in \cite{miniCAD} (\textcolor{red}{Long Version too}). 

After recalling relevant material from \cite{miniCAD,miniCAD-extended} in Section \ref{sec:background}, we present in Section~\ref{sec:technical} several preliminary technical result, which are then used in the proof of the main theorem given in Section~\ref{sec:proof}.
%The proof combines topological arguments with techniques developed in \cite{miniCAD} (\textcolor{red}{Long Version too}).

\section{Background and Notation}\label{sec:background}

We use the notation of \cite{exotic}, which we reproduce here in Subsection~\ref{sub:CAD} to be self-contained.
%In Subsection \ref{sub:CAD}, we fix the notation as in \cite{exotic}, which we reproduce here to be self-contained. 
Then, we recall in the next subsection the main ingredients of \cite{miniCAD,miniCAD-extended} for studying the existence of minimum CADs. We refer the reader to those paper for more details.

\subsection{Cylindrical Algebraic Decompositions}\label{sub:CAD}

For every positive integer $n \in \mathbb{N}^*$, we consider CADs of $\mathbb{R}^n$ with respect to a fixed variable ordering. Since the cells of such CADs are defined inductively from cells of CADs of $\mathbb{R}^k$ ($k \leq n$) which are indexed by $k$-tuples, it is convenient to adopt the tuple notation of \cite{miniCAD}.
More precisely, we identify any $k$-tuple $I=(i_1,\ldots,i_k)\in\mathbb{N}^k$ with the corresponding word $i_1\ldots i_k$, and we set $|I| = k$ to be its length. We say that $I$ is odd (resp. even) if $i_k$ is odd (resp. even). We denote by $\varepsilon$ the empty tuple, which corresponds to the empty word. For $j\in \mathbb{N}$, we denote by $I:j$ the $k+1$ tuple $(i_1,\ldots,i_k,j)$.

We also use the $\odot$ shorthand of \cite{binyamini2019complex} to deal easily with sectors and sections, defined below. For every subset $S \subseteq \R^n$, we denote (following the notation of \cite{bochnaketal1998}) by $\mathcal{S}^0(S)$ the ring of continuous semi-algebraic functions from $S$ to $\R$.
For $f \in \mathcal{S}^0(S)$, the section $S \odot \{f\}$ with base $S$ and bound $f$ is simply the graph of $f$, that is
    $$S \odot \{f\} = \{(\textbf{x},y)\in\R^{n+1} \; | \; \textbf{x} \in S,  y = f(\textbf{x})\}.$$
For $l \in \mathcal{S}^0(S) \cup \{-\infty\}$ and $u \in \mathcal{S}^0(S) \cup \{+ \infty\}$ such that $l < u$ on $S$, the sector $S \odot (l,u)$ with base $S$, lower bound $l$ and upper bound $u$, is the set
    $$S \odot (l,u) = \{(\textbf{x},y)\in\R^{n+1} \; | \; \textbf{x} \in S, l(\textbf{x}) < y < u(\textbf{x})\}.$$ 
These definitions naturally extend to $n = 0$, where $\R^0 = \{0\}$. For instance, a real algebraic number $\xi$ is identified with the constant semi-algebraic function $0 \mapsto \xi$ and we set $\{\xi\} = \R^0 \odot \{\xi\}$.

%\lucas{A-t-on besoin de ceci ?}Similarly to \cite{DLSregular}, we denote by $\overline{S}$ the closure of $S$ in the Euclidean topology, and by $\partial S = \overline{S} \setminus S$ its (cell) boundary.  

\begin{definition}\label{def:cad}
    A cylindrical algebraic decomposition ($\CAD$) of $\mathbb{R}^n$ is a sequence $\mathscr{C} = (\mathscr{C}_1,\ldots,\mathscr{C}_n)$ such that  
    for all $k \in \{1,\ldots,n\}$, the set  $$\mathscr{C}_k= \left\{C_{i_1 \cdots i_k} |  \forall j \in \{1,\ldots,k\}, i_j \in \{1,\ldots,2u_{i_1\cdots i_{j-1}} +1\}\right\}$$
     is a finite semi-algebraic partition of $\mathbb{R}^k$  defined inductively by the following data:
    \begin{enumerate}%[leftmargin=0.35cm]
        \item $\Cr_0 = \{C_\varepsilon\}$ where $C_\varepsilon = \{0\}$;
        \item for each cell $C_I \in \mathscr{C}_k$ ($k < n$), there exist a natural number $u_I \in \mathbb{N}$ and $\xi_{I: 2 },\xi_{I:4},\ldots,\xi_{I:2 u_I} \in \mathcal{S}^0(C_I)$ (possibly none if $u_I = 0$) with
         $\xi_{I: 2 } < \xi_{I:4} < \ldots <\xi_{I:2 u_I}$ on $C_I$, that define exactly all cells of $\mathscr{C}_{k+1}$ by
        {\small\begin{align*}
             C_{I:2j} &= C_I \odot \{\xi_{I:2j}\} , \quad &(1 \leq j \leq u_I)\\
             C_{I:2j+1} &= C_I \odot (\xi_{I:2j},\xi_{I:2(j+1)}),\quad &(0 \leq j \leq u_I)
         \end{align*}}
         with the convention that $\xi_{I:0} = -\infty$ and $\xi_{I:2u_I+ 2} = +\infty$.
    \end{enumerate}
 We say that the element $C_{I}$ of $\mathscr{C}_k$ is a CAD cell of index $I$.  
\end{definition}

For all $k \in \{0, \ldots, n\}$, we denote by $\pi_k : \R^n \to \R^{k}$ the projection onto the first $k$ coordinates, defined by $\pi_k(x_1, \ldots, x_n) = (x_1, \ldots, x_k)$, and use the same notation for its natural successive extensions to subsets of $\R^n$ and to families of subsets of $\R^n$.
In particular, for every $\CAD$ $\Cr$ of $\R^n$, we have $\pi_k(\Cr_n) = \Cr_k$ for every $ k\in \{0, \ldots, n\}$. For this reason, we usually identify $\Cr$ with $\Cr_n$ when the context is~clear.

In general, one can refer to a CAD cell $C$ of $\R^n$ without referring to an ambient CAD, nor its index. The unique CAD cell of $\R^0$ is $\{0\}$ and if $n \geqslant 1$, then $\pi_{n-1}(C)$ is a CAD cell of $\R^{n-1}$, and $C$ is a section or a sector with base $\pi_{n-1}(C)$ as defined above.

\subsection{Minimal and minimum CADs}\label{sub:miniCAD}

Following \cite{miniCAD, miniCAD-extended}, we say that a CAD $\mathscr{C}$ is adapted to a given semi-algebraic set $S$ of $\R^n$ if $S$ is a union of cells of $\mathscr{C}$. We denote by $\text{CAD} (S)$ the set of all CADs adapted to $S$. If $\mathcal{F} = (S_1, \ldots, S_p)$ is a finite family of semi-algebraic sets of $\R^n$, then we consider the set $$\CAD(\Fc) = \bigcap_{i=1}^p\CAD(S_i)$$
of CADs adapted to $\Fc$.

We use the refinement order defined on the set of partitions to compare CADs. More precisely, if $\mathscr{C}$ and $\mathscr{C}'$ are two CADs of $\mathbb{R}^n$, then we say that $\mathscr{C}$ is smaller than or equal to $\mathscr{C}'$ if every cell of $\mathscr{C}$ is a union of cells of $\mathscr{C}'$. In this case, we write $\mathscr{C} \preceq \mathscr{C}'$ or $\mathscr{C}' \succeq \mathscr{C}$.

\begin{definition}
        A minimal CAD adapted to $\Fc$ is a minimal element of the poset ($\text{CAD}(\Fc),\preceq$). 
        % In other words, $\mathscr{C} \in \text{CAD}(S)$ is a minimal CAD adapted to $S$ if 
        % \[\forall \mathscr{C}' \in \text{CAD}(S), (\mathscr{C}'\preceq \mathscr{C})  \implies (\mathscr{C}' =  \mathscr{C}).\] 
        A minimum CAD adapted to $\Fc$ is a minimum element of ($\text{CAD}(\Fc),\preceq$). 
    %     Namely,  $\mathscr{C} \in \text{CAD}(S)$ is a minimum CAD adapted to $S$ if 
    % \[\forall \mathscr{C}' \in \text{CAD}(S), \mathscr{C}'\succeq \mathscr{C}.\]     
        We say that $\Fc$ admits a minimum CAD if there exists a minimum CAD adapted to $\Fc$.
    \end{definition}

%\subsection{Rewritings of CADs}

To study and compute minimal and minimum CADs, it is convenient to work with CAD reductions, which we recall in Definition~\ref{def:CADRed} (see Section 5 of \cite{miniCAD-extended} for a detailed presentation).

Roughly speaking, a reduction $\Phi_A$ (defined below) of a CAD $\Cr$ of $\R^n$ consists in removing a section $C_A \in \Cr_k$ (where $k \in \{1, \ldots, n\}$). To do so, one merges the section $C_A$ with the sectors immediately below and above of $\Cr_k$, namely $C_{A - e_k}$ and $C_{A + e_k}$, and then merges the corresponding cells of $\Cr_{k+1}, \ldots, \Cr_n$ in the cylinders lying above.
In order for this operation to yield a valid CAD, one may first argue at a combinatorial level. For instance, one checks that the number of cells in $\Cr_{k+1}$ projecting onto $C_A$ coincides with the number of cells in $\Cr_{k+1}$ projecting onto each of $C_{A \pm e_k}$. 
To this end, we encode the relevant combinatorial structure of a CAD as a tree.
\begin{definition}\label{def:treeCAD}
    For every $\mathscr{C} \in \text{CAD}(\Fc)$, the CAD tree associated with $(\mathscr{C}, \Fc)$ is the pair $(T,L)$ where $T$ is the prefix tree of depth $n$ whose nodes are the indices of the cells of $\mathscr{C}$, and where $L$ is a map defined recursively on $T$ from the leaves to the root as follows: if $I$ is a leaf, then
    $L(I) \in \{0,1\}^p$, with the $i^\text{th}$ component equals to 1 if and only if $C_I \subseteq S_i$; if $I$ is not a leaf, $L(I) = \big(L(I:1), \ldots, L(I:2u_I+1)\big)$. We denote this pair by Tree$(\mathscr{C} ,\Fc)$.
\end{definition}

If $C_I \in \Cr_k$ with $k < n$, then $L(I)$ encodes the arrangement of the cells of $\Cr_{k+1}, \ldots,\Cr_{n}$ that lie in a cylinder projecting onto $C_I$. For instance, the length of $L(I)$ equals the number of cells in $C_{k+1}$ contained in the cylinder $C_I \times \R$. 

If the cells in the cylinders above $C_{A-e_k}$, $C_A$, and $C_{A+e_k}$ can be merged to obtain a coarser CAD $\Cr'$, then some cells of $\Cr$ are merged and the indices of some other are modified.
The relabelling map $\psi_A$, defined below, sends the index $I$ of a CAD cell $C_I \in \Cr_l$ (with $l \in \{1, \ldots, n\}$) before merging to the index of the corresponding cell in $\Cr'_l$ after merging. The auxiliary sets $S_A$ and $N_A$ appearing in the definition consist precisely of the indices of the cells affected by this relabelling.

\begin{definition}\label{def:auxpsi} 
For every $k \in \N^*$, we define the $k^\text{th}$ prefix map by
    \[p_{k} \colon \bigcup_{l=0}^{+\infty}(\mathbb{N}^*)^l \to \bigcup_{l=0}^{+\infty}(\mathbb{N}^*)^l \colon (i_1,\ldots,i_l) \mapsto \begin{cases}
        (i_1,\ldots,i_l)&\text{ if $l < k$},\\
        (i_1,\ldots,i_k)&\text{ if $l\geq k$}.
       \end{cases}
    \]
    For every tuple $A \in (\mathbb{N}^*)^k$ ($k \geq 1$),  we define the sets%\footnote{\textcolor{red}{Intuitively, $S_A$ stands for the sons of $A$ (more precisely, $A$ and its descendants), $N_A$ stands for some nephews of $A$ (the older siblings of $A$ and their descendants) and $F_A$ for the rest of the family. }}
    \begin{align*}
        S_A &= \left\{I \in \bigcup_{l=0}^{+\infty}(\mathbb{N}^*)^l : p_k(I) = A\right\},\\
        N_A &= \left\{I \in \bigcup_{l=0}^{+\infty}(\mathbb{N}^*)^l : \exists m \in \N^* :  p_k(I) = A + me_k\right\},\\
        F_A &= \bigcup_{l=0}^{+\infty}(\mathbb{N}^*)^l \setminus \left(S_A \cup N_A\right).
    \end{align*}
    where $e_k$ is the $k^\text{th}$ unit vector of $\mathbb{R}^l$ for $l \geq k$. If $A$ is even, we also define the relabelling map
    {\begin{align*}
        \psi_A \colon \bigcup_{l=0}^{+\infty}(\mathbb{N}^*)^l \to \bigcup_{l=0}^{+\infty}(\mathbb{N}^*)^l \colon I \mapsto \begin{cases}
            I - e_k &\text{ if } I \in S_A,\\
            I - 2 e_k &\text{ if } I \in N_A,\\
            I &\text{ if } I \in F_A.
        \end{cases}
    \end{align*}}
\end{definition}

\begin{definition}\label{def:CADtreeRed}
    Consider a CAD tree $\mathcal{T} = (T,L)$ and an even node $A \in T \cap (\mathbb{N}^*)^k$ ($k \geq 1$). We say that $\psi_A$ induces a reduction rule on $\mathcal{T}$ if we have
    \begin{equation}\label{eqn:L}
             L(A - e_k) = L(A) = L(A + e_k).
         \end{equation}
 Then the induced reduction rule is denoted by $\Psi_{A}$ and the reduced CAD tree is given by $\mathcal{T}' =  (\psi_{A}(T), L')$ where $L'$ is defined on the leaves of $\psi_A(T)$ by $L'(\psi_A(I))=L(I)$ for every leaf $I$ of $T$. In this case, we write $\mathcal{T} \to \mathcal{T'}$ or $\mathcal{T}' \leftarrow \mathcal{T}$.
 \end{definition}

 \begin{definition}\label{def:CADRed}
    Let $\Cr \in \CAD(\Fc), \Tree(\Cr, \Fc) = (T,L)$ and $A$ be an even node of $T$. We suppose that $\Psi_A$ is a tree reduction rule defined from $\Tree(\Cr,\Fc)$ to $\mathcal{T}' =(T',L')$ and we write
    $$\mathscr{C}' = \left\{\bigcup_{I \in T \; : \; \psi_A(I) = I'} C_{I} \; \Big| \; I' \text{ leaf of } T'\right\}.$$ 
    We say that $\Psi_A$ lifts to a $\CAD(\Fc)$ reduction rule $\Phi_A$ defined on $\mathscr{C}$ if $\Cr'$ is a CAD (adapted to $\Fc$). In this case, we say that $\Cr'$ is the reduced CAD and we write $\mathscr{C} \to \mathscr{C}'$ or $\mathscr{C}' \leftarrow \mathscr{C}$.
\end{definition}

We denote by $\stackrel{*}{\leftarrow}$ the reflexive and transitive closure of $\leftarrow$ on $\CAD(\Fc)$.
Theorem 5.9 of \cite{miniCAD-extended} asserts that the relation $\stackrel{*}{\leftarrow}$ is equal to $\preceq$ on $\CAD(\Fc)$.

\section{Technical results}\label{sec:technical}
% We now develop all the machinery needed to carry out the proof of the main result. We still consider $\Fc = (S_1, \ldots, S_p)$ a finite family of semi-algebraic sets of $\R^n$ ($n \in \N^*$). 
% \lucas[inline]{The idea of Theorem \ref{thrm:main-curtained} is to impose additional assumptions on the sets defining $\mathcal{F}$ to ensure the existence of a minimum CAD.
% In Subsection \ref{sub:CADnirr}, we provide a characterization of the existence of a minimum CAD based on the confluence of a strict poset of $\CAD(\Fc)$.}
In this section, we develop the machinery required for the proof of the main result (Section \ref{sec:proof}). We still consider a finite family $\Fc = (S_1, \ldots, S_p)$ of semi-algebraic subsets of $\R^n$, with $n \in \N^*$.

%The main idea underlying Theorem~\ref{thrm:main-curtained} is to impose additional assumptions on the sets defining $\Fc$ in order to ensure the existence of a minimum CAD. 

In Subsection \ref{sub:curtained}, we identify two properties (closedness and curtainedness) that fail in the counterexamples of \cite{miniCAD}. We will show in Section \ref{sec:proof} that these properties are actually sufficient to guarantee the existence of a minimum adapted CAD (when $n = 3$).
In Subsection~\ref{sub:CADnirr}, we provide a characterization of this existence in terms of the confluence (a property based on $\CAD(\Fc)$ reductions) of a strict subset of $\CAD(\Fc)$.
The remaining subsections provide the technical ingredients needed for the proof of the main result based on that characterization. In Subsection~\ref{sub:liftingTreeRed}, we give a criterion ensuring that a tree reduction lifts to a CAD reduction, based on the continuity of certain functions defined by unions. In Subsection~\ref{sub:pasting}, we establish conditions guaranteeing the continuity of such functions. Finally, in Subsection~\ref{sec:lbc}, we prove some topological properties of CAD cells of the plane, which are then used to verify that these continuity conditions are satisfied in Section~\ref{sec:proof} under the closedness and curtainedness assumptions.

\subsection{Closed and curtained sets}\label{sub:curtained}
We motivate the assumptions of Theorem \ref{thrm:main-curtained} by observing some properties that are not satisfied in the counterexamples of \cite{miniCAD}. First, the Trousers $\mathbb{T}$ (see Definition~4.3 of \cite{miniCAD}) is not closed, and therefore cannot be expressed as a finite union of sections of a continuous function (see the Closed Graph Theorem: Proposition~2.14 of \cite{rudin-FA}). Note that this explains the non-confluence of $\CAD(\mathbb{T})$ discussed in Subsection~5.3 of \cite{miniCAD}. However, closedness alone is not sufficient, as illustrated by Example~4.9 of \cite{miniCAD}, where certain closed sets intersect a vertical line $\ell = \{p\} \times \R \subset \R^3$ in infinitely many points without coinciding with $\ell$. This phenomenon is ruled out by the curtained condition introduced in Definition \ref{def:curtained} below. As we prove in Theorem~\ref{thrm:main-curtained}, assuming that all sets defining $\Fc$ are closed and curtained ensures the existence of a minimum CAD (when $n = 3$).

The geometric notion of curtains frequently appears in the study of CADs and related algorithms (see \cite{curtains} and the references therein).
%\textcolor{blue}{Note that we only consider curtains in the last direction.}
%In the context of this paper, this assumption naturally arises as a convenient additional hypothesis (together with closedness) to ensure the existence of a minimum CAD.% (irreducible at the last level).

\begin{definition}\label{def:curtained}
    Let $S$ be a semi-algebraic set of $\R^n$ and $W \subseteq \R^{n-1}$. We say that $S$ has a curtain (in the last direction) at $W$ if $W \times \R \subseteq S$. 
    
    We say that $S$ is curtained if for every $p \in \R^{n-1}$, the fibre $S \cap \left( \{p\} \times \R\right)$ is either finite, or equal to the whole line $\{p\} \times \R$ (i.e. $S$ has a curtain at $\{p\}$).
\end{definition}

\begin{proposition}\label{prop:curtained-algebraic}
    Every algebraic set $V$ in $\R^n$ is closed and curtained. 
\end{proposition}
The converse is not true: roughly speaking, most of the closed and curtained semi-algebraic sets are not algebraic.
\begin{proof}
    By definition, an algebraic set $V \subseteq \R^n$ is the common zero locus of finitely many $n$-variate polynomials with real coefficients $P_1,\dots,P_K \in \R[X_1,\dots,X_n]$. Since each $P_k$ is continuous,
    their common vanishing set $V$ is closed.

    Consider now $p = (x_1,\dots,x_{n-1}) \in \R^{n-1}$ such that the fibre $V \cap \left( \{p\} \times \R \right)$ is infinite. For each $k \in \{1,\dots,K\}$, consider the univariate polynomial
    $P_k(x_1,\dots,x_{n-1},X_n) \in \R[X_n]$.
    Since it has infinitely many real roots (in $X_n$), it must be the zero polynomial. Hence
    $P_k$ vanishes identically on $\{p\} \times \R$ for every $k$, and therefore
    $\{p\} \times \R \subseteq V$. This shows that $V$ is curtained.
\end{proof}
Thus, the class of closed and curtained sets contains the class of algebraic sets. Moreover, when $n = 3$, it enjoys useful topological properties of specific interest here (see Corollary~\ref{cor:bypass} below).
%The next result concerning the closedness and curtainedness assumptions is delayed to Corollary \ref{cor:bypass}.

\subsection{CADs irreducible at the last level}\label{sub:CADnirr}
We introduce a strict subset $\CAD_{n\text{-irr}}(\Fc)$ of $\CAD(\Fc)$ that shares the same minimal elements, and in which every section of the CADs under consideration is directly related to the sets defining $\Fc$. The first property will reduce the discussion in the proof of the main theorem, while the second is crucial for exploiting the assumptions on $\Fc$, specifically in Subsection \ref{case:|A|neq|B|}.

To obtain this subset, we intuitively start from $\CAD(\Fc)$ and perform some trivial $\CAD(\Fc)$ reductions, namely those occurring at the last level.
More precisely, for every $\Cr \in \CAD(\Fc)$, we observe that every tree reduction rule $\Psi_I$ defined on $\Tree(\Cr, \Fc)$ at the last level (i.e. with $|I| = n$) lifts to a $\CAD(\Fc)$ reduction rule $\Phi_I$ defined on $\Cr$. In particular, no additional topological assumptions are required, in contrast to the case where $|I| < n$. %This observation leads us to the study of CADs adapted to $\Fc$ for which no reduction is defined at the last level.

\begin{definition}
    A CAD $\Cr \in \CAD(\Fc)$ is said to be irreducible at the last level (or $n$-irreducible for short) if there exists no $\CAD(\Fc)$ reduction $\Phi_I$ defined on $\Cr$ with $|I| = n$. We denote by $\CAD_{n\text{-irr}}(\Fc)$ the set of all $n$-irreducible CADs adapted to $\Fc$.
\end{definition}
This condition may be viewed as a generalization, to finite families of semi-algebraic sets, of the notion of reduced CAD %(not to be confused with the notion of CAD reductions of \cite{miniCAD})
that can be found in the literature in Definition 3.1.15 of \cite{locatelli}, or as part of the definition of a basis-determined CAD of~\cite[p. 34]{arnon}). 

We now gather some properties about $\CAD_{n\text{-irr}}(\Fc)$.

%The advantages of considering $n$-irreducible CADs adapted to $\Fc$ in the study of minimal and minimum CADs are at least twofold. 
%First, the sections of such CADs are directly related to the sets $S_i$ under consideration. This fact will play a crucial role in the proof of the main theorem (in Subsection \ref{case:|A|neq|B|} below).
%Second, although the inclusion $\CAD_{n\text{-irr}}(\Fc) \subsetneq \CAD(\Fc)$ is strict in general, we show that any minimal CAD adapted to $\Fc$ is necessarily irreducible at the last level.

\begin{proposition}\label{prop:minin-irr}
    A CAD $\Cr$ in $\R^n$ is a minimal element of $\CAD(\Fc)$ if and only if it is a minimal element of $\CAD_{n\text{-irr}}(\Fc)$. In particular, the poset $(\CAD(\Fc), \preceq)$ admits a minimum if and only if the poset $(\CAD_{n\text{-irr}}(\Fc), \preceq)$ does.
\end{proposition}
The proof relies on the following lemma together with elementary set-theoretic arguments.
\begin{lemma}\label{lemma:thanks}
If $\Cr \in \CAD(\Fc)$, then there exists $\Cr' \in \CAD_{n\text{-irr}}(\Fc)$ such that $\Cr' \preceq \Cr$. In particular, $\CAD_{n\text{-irr}}(\Fc)$ is never empty. 
\end{lemma}
\begin{proof}
    This directly follows from the fact that only finitely many $\CAD(\Fc)$ reductions can be performed.
\end{proof}

Every CAD smaller than a $n$-irreducible one is also $n$-irreducible.
\begin{lemma}\label{lemma:smallerRed}
    Let $\Cr \in \CAD(\Fc)$ and $\Dr \in \CAD_{n\text{-irr}}(\Fc)$. If $\Cr \preceq \Dr$, then $\Cr \in \CAD_{n\text{-irr}}(\Fc)$. %In particular, if $\Cr \leftarrow \Dr$, then $\Cr \in \CAD_{n\text{-irr}}(\Fc)$.
\end{lemma}
\begin{proof}
    %For the first item, we set $\Cr'$ to be a minimal element of CAD$(\Fc)$ smaller than or equal to $\Cr$ (see Proposition 3.1 of \cite{miniCAD-extended}). Such $\Cr'$ is $n$-irreducible, since otherwise it is not minimal.
    We suppose for the sake of a contradiction, that $\Cr$ is not $n$-irreducible. In this case, there exists a CAD$(\Fc)$ reduction $\Phi_I$ defined on $\Cr$ with $|I| = n$. 
    By assumption, $C_I \in \Cr_n$ is a section that is a cell, or a union of cells, of $\Dr_n$, which are necessarily sections. Let $D_J \in \Dr_n$ be one of them. It is straightforward to show that $\Phi_J$ is a CAD$(\Fc)$ reduction defined on $\Dr$ (at the last level since $|J| = n$), which is a contradiction.
\end{proof}

We obtain the following result, which is similar to Theorem 5.9 of \cite{miniCAD}, but in the $n$-irreducible framework. Note that Lemma~\ref{lemma:smallerRed} guarantees that if $\Cr, \Dr \in \CAD_{n\text{-irr}}(\Fc)$, then we have $\Cr \stackrel{*}{\leftarrow} \Dr$ if and only if either $\Cr = \Dr$, or there exists $\Er \in \CAD_{n\text{-irr}}(\Fc)$ such that $\Cr \stackrel{*}{\leftarrow} \Er \leftarrow \Dr$ (i.e. the chain of reduction stays in $\CAD_{n\text{-irr}}(\Fc)$).
\begin{lemma}\label{lemma:thrm59Red}
    Let $\Cr, \Dr \in \CAD_{n-\text{irr}}(\Fc)$. We have $\Cr \stackrel{*}{\leftarrow} \Dr$ if and only if $\Cr \preceq \Dr$.
\end{lemma}

As in Subsection 5.3 of \cite{miniCAD}, the reduction system $\CAD_{n-\text{irr}}(\Fc)$ is said to be confluent if for every $\mathscr{C}, \mathscr{D},  \mathscr{E}  \in \CAD_{n-\text{irr}}(\Fc)$, we have
% \begin{align*}
%     \left(\mathscr{C} \stackrel{*}{\leftarrow}  \overline{\mathscr{C}} \stackrel{*}{\to} \mathscr{C}'\right) 
%     \implies \left(\exists \underline{\mathscr{C}} \in \CAD_{n-\text{irr}}(\Fc) : \mathscr{C} \stackrel{*}{\to} \underline{\mathscr{C}} \stackrel{*}{\leftarrow} \mathscr{C}'\right)
% \end{align*}
\begin{align*}
    \left(\mathscr{D} \stackrel{*}{\leftarrow} \mathscr{C} \stackrel{*}{\to} \mathscr{E}\right) 
    \implies \left(\exists \mathscr{F} \in \CAD_{n-\text{irr}}(\Fc) : \mathscr{D} \stackrel{*}{\to} \mathscr{F} \stackrel{*}{\leftarrow} \mathscr{E}\right).
\end{align*}
Since there exists no infinite chain of reductions in $\CAD_{n-\text{irr}}(\Fc)$, we can show that $\CAD_{n-\text{irr}}(\Fc)$ is confluent if and only if for every $\mathscr{D}, \mathscr{C}, \mathscr{E}  \in \CAD_{n-\text{irr}}(\Fc)$, we have
\begin{align*}
    \left(\mathscr{D} \leftarrow  \mathscr{C} \to \mathscr{E}\right) 
    \implies \left(\exists \mathscr{F} \in \CAD_{n-\text{irr}}(\Fc) : \mathscr{D} \stackrel{*}{\to} \mathscr{F} \stackrel{*}{\leftarrow} \mathscr{E}\right).
\end{align*}
The following result, reminiscent to Theorem 5.16 of \cite{miniCAD-extended}, plays a crucial role in the proof of the main theorem.

% $\mathscr{C}, \overline{\mathscr{C}}, \mathscr{C}'  \in \CAD_{n-\text{irr}}(\Fc)$ such that $\mathscr{C} \stackrel{*}{\leftarrow} \overline{\mathscr{C}} \stackrel{*}{\to} \mathscr{C}'$, there exists $\underline{\mathscr{C}} \in \CAD_{n-\text{irr}}(\Fc)$ such that $\mathscr{C} \stackrel{*}{\to} \underline{\mathscr{C}} \stackrel{*}{\leftarrow} \mathscr{C}'$. Since there exists no infinite chain of reductions in $\CAD_{n-\text{irr}}(\Fc)$, we can show that $\CAD_{n-\text{irr}}(\Fc)$ is globally confluent if and only if if for all $\mathscr{C}, \overline{\mathscr{C}}, \mathscr{C}'  \in \CAD_{n-\text{irr}}(\Fc)$ such that $\mathscr{C} \stackrel{}{\leftarrow} \overline{\mathscr{C}} \stackrel{}{\to} \mathscr{C}'$, there exists $\underline{\mathscr{C}} \in \CAD_{n-\text{irr}}(\Fc) : \mathscr{C} \stackrel{*}{\to} \underline{\mathscr{C}} \stackrel{*}{\leftarrow} \mathscr{C}'$.
 
\begin{proposition}\label{prop:conflRed}
    The poset $(\CAD(\Fc), \preceq)$ admits a minimum if and only if the abstract reduction system $(\CAD_{n-\text{irr}}(\Fc), \leftarrow)$ is confluent.
\end{proposition}
    Thanks to Proposition \ref{prop:minin-irr} and Lemma \ref{lemma:thanks}, the proof is analogous to that of Theorem 5.16 of \cite{miniCAD-extended} but in the $n$-irreducible framework.

\subsection{Towards practical lifting of tree reductions}\label{sub:liftingTreeRed}
%The proof of Theorem \ref{thrm:main} in Section \ref{sec:proof} relies on Proposition \ref{prop:conflRed}.  
In order to apply Proposition \ref{prop:conflRed} in Section \ref{sec:proof}, we characterize in Proposition \ref{prop:liftingReduction} below the conditions under which a tree reduction rule $\Psi_A$ lifts to a CAD reduction rule  $\Phi_A$ (see Definition \ref{def:CADRed}), in terms of the continuity of certain functions obtained via unions.

Let $f_1, f_2$ be two functions defined on $X_1, X_2 \subseteq \R^n$ respectively. Recall that if both functions agree on $X_1 \cap X_2$, then the union $f_1 \cup f_2$ is another function defined on $X_1 \cup X_2$ by $f_1 \cup f_2|_{X_i} = f_i$ for $i \in \{1,2\}$. Note that the condition is satisfied when $X_1 \cap X_2 = \emptyset$.

%Recall that, in set-theoretic terms, a function $f$ defined on $X \subseteq \R^n$ (written $f : X \to \R$) is a binary relation between $X$ and $\R$ (i.e. $f \subseteq X \times \R$) such that for every $x \in X$, there exists a unique $y \in \R$ such that $(x,y) \in f$. This unique $y$ is usually denoted $f(x)$, and the set $X$ is called the domain of $f$. If $g$ is another function defined on $X' \subseteq \R^n$, then $f \cup g$ is a binary relation between $X \cup X'$ and $\R$. %The graph of $f$ is simply the set $f$, and is denoted by $X \odot f$ here to avoid any confusion.
%The union $f \cup g$ is a function if and only if $f$ and $g$ coincide on $X \cap X'$, i.e. if for every $x \in X \cap X'$, we have $f(x) = g(x)$. In particular, if $X$ and $X'$ are disjoint, then $f \cup g$ is a function. 

We begin with two preliminary lemmas. The first one concerns the relabelling function $\psi_A$ and the sets $S_A, N_A$ and $F_A$ (see Definition \ref{def:auxpsi}). This lemma will be used repeatedly throughout the paper, for instance to analyse the sets of $\Cr'$ and the functions $\xi'_{I'}$ in the proof of Proposition \ref{prop:liftingReduction} below.

\begin{lemma}\label{lemma:psiBij}
    For every even tuple $A \in (\mathbb{N}^*)^k$ ($k \geq 1$), the restriction of $\psi_A$ to any of the sets $S_A, N_A, F_A$ is a bijection from each set to its image, given by 
    \begin{align}
    \begin{cases}
        \psi_A(S_A) = S_{A-e_k},\\\psi_A(N_A) = S_{A-e_k} \cup S_A \cup N_A, \\
        \psi_A(F_A) = F_A. 
    \end{cases}
    \end{align}
    Moreover, for every $I \in \bigcup_{l=0}^{+\infty}(\mathbb{N}^*)^l$, we have
    \begin{align*}
        (\psi_A)^{-1}(\{I\}) = \begin{cases}
            \{I, I + e_k, I + 2e_k\} \text{ if } I \in S_{A-e_k},\\
            \{I + 2e_k\} \text{ if } I \in S_{A} \cup N_A,\\
            \{I\} \text{ otherwise.}
        \end{cases}
    \end{align*}
\end{lemma}
\begin{proof}
    The first part of the lemma is straightforward. For the second part, we first observe that we have a partition $(\N^*)^n = S_A \sqcup N_A \sqcup F_A$ (where $\sqcup$ denotes a disjoint union), and that $S_{A-e_k}$ is a subset of $F_A$. Hence, every $I \in S_{A-e_k}$ is the image of a unique element in each of the sets $F_A, S_A$ and $N_A$. These elements are respectively $I, I+e_k$ and $I + 2e_k$. If $I \in S_A \cup N_A$, it is not in $S_{A-e_k}$, nor in $F_A$, and is therefore the image of a single element in $N_A$, namely $I + 2e_k$. Finally, if $I \in F_A$, then it is a fixed point of $\psi_A$.
\end{proof}

The definition of a CAD (see Definition \ref{def:cad}) can be directly rephrased to yield the following convenient characterization. 

\begin{lemma}\label{lemma:CADcharact}
    Let $\Dr_n$ be a partition of $\R^n$ and set $\Dr_{k} = \pi_{k}(\Dr_n)$ for every $k \in \{0, \ldots, n\}$. The tuple $\Dr = (\Dr_1, \ldots, \Dr_n)$ is a CAD if and only if $\Dr_{n-1}$ is a CAD of $\R^{n-1}$ and every $D \in \Dr_n$ is a section or a sector with basis $\pi_{n-1}(D)$, and $\pi_{n-1}(D) \in \Dr_{n-1}$.
\end{lemma}
Under the assumptions and with the notation of Definition \ref{def:CADRed}, the following holds.
\begin{proposition}\label{prop:liftingReduction}
    For every even $I' = J : 2j \in T'$, the union $$\xi'_{I'} = \bigcup_{I \in T \; : \; \psi_A(I) = I'} \xi_{I}$$ is a function defined on $C'_J$.
    The following assertions are equivalent:
    \begin{enumerate}[(i)]
        \item The tree reduction rule $\Psi_A$ lifts to a CAD($\Fc$) reduction rule $\Phi_A$ defined on $\Cr$ and the corresponding reduced CAD is  $\Cr'$; \label{item:(i)}
        \item The finite partition $\Cr'$ is a CAD; \label{item:(ii)}
        \item For every even $I' \in T'$, the function $\xi'_{I'}$ is continuous; \label{item:(iii)}
        \item For every even $I' \in T' \cap S_{A-e_{|A|}}$, the function $\xi'_{I'}$ is continuous.\label{item:(iv)}
    \end{enumerate}
\end{proposition}

\begin{proof}
 %We start with a general observation before proving the equivalences. 
 For every even $I' = J : 2j \in T'$, Lemma \ref{lemma:psiBij} and the definition of $\xi'_{I'}$ yield
    \begin{align}\label{eqn:xi'I'}
        \xi'_{I'} = 
        \begin{cases}
            \xi_{I'} \cup \xi_{I'+ e_{|A|}} \cup \xi_{I' + 2 e_{|A|}} &\text{ if } I' \in S_{A - e_{|A|}}, \\
            \xi_{I' + 2 e_{|A|}} &\text{ if } I' \in S_A \cup N_A, \\
            \xi_{I'} &\text{ otherwise.}
        \end{cases}
    \end{align}
Considering each case separately, and applying the definition of $C'_J$, we obtain that $\xi'_{I'}$ is a function defined on $C'_{J}$.
For instance, writing $A = (a_1, \ldots, a_{|A|})$, if $I' \in S_A \cup N_A$ and if $|I'| > |A|$, then there exists $m \in \N$ such that 
$I' = J:2j$ with $$J = (a_1, \ldots, a_{|A|}+m, i'_{|A|+1}, \ldots, i'_{|I'|-1}).$$
%$$I' = \underbrace{(a_1, \ldots, a_{|A|}+m, i'_{|A|+1}, \ldots, i'_{|I'|-1})}_{J=} :  2j.$$
In this case, $J \in S_A \cup N_A$ and thus $C'_J = C_{J+2e_{|A|}}$, which is indeed the domain of $\xi_{I'+2e_{|A|}}$. The proof of the other cases are analogous.

The equivalence of the first and second items is precisely the definition of CAD reduction rules (see Definition \ref{def:CADRed}).
To show that \ref{item:(ii)} implies \ref{item:(iii)}, we observe that if $\Cr'$ is a CAD, then for every even $I' = J : 2j \in T'$ of length $l \in \{2, \ldots, n\}$, the CAD cell $C'_{I'} \in \Cr'_{l}$ is the section $C'_J \odot \{\xi'_{I'}\}$. By definition of a CAD, the function $\xi'_{I'}$ is continuous.   
The fact that \ref{item:(iii)} implies \ref{item:(ii)} follows from Lemma \ref{lemma:CADcharact}, and the following two observations. First, if $l = |A|$, then we have
$$\mathscr{C}'_l  = \left(\Cr_l \setminus\{C_{A - e_{l}},C_{A}, C_{A - e_{l}}\}\right) \cup \{C_{A - e_{l}} \cup C_{A} \cup C_{A - e_{l}}\},$$
which is a CAD of $\R^l$.
Second, for every $l \in \{1, \ldots, n\}$ and $ C'_{I'} \in \Cr'_l$
we have
            \begin{align*}
        C'_{I'} = \begin{cases}
           C'_J \odot \{\xi'_{J:2j}\} &\text{ if } I' = J:2j, \\
            C'_J \odot(\xi'_{J:2j}, \xi'_{J:2(j+1)}) &\text{ if } I' = J:2j+1.
        \end{cases}
    \end{align*}
The fourth assertion is readily implied by the third. The converse is also true since for every even $I' \in T' \setminus S_{A-e_{|A|}}$, the function $\xi'_{I'}$ coincides either with $\xi_{I'+2e_{|A|}}$ or with $\xi_{I'}$ (see Equation \eqref{eqn:xi'I'}), which are both continuous since $\Cr$ is a CAD.     
\end{proof}

\subsection{Pasting lemmas}\label{sub:pasting}

To apply Proposition \ref{prop:liftingReduction} in Section \ref{sec:proof}, it will be convenient to show that some unions of functions are continuous. 

In this framework, an elementary result in Topology is the pasting lemma (see Theorem 18.3 of \cite{munkres} for a proof).

\begin{lemma}[The pasting lemma]\label{lemma:classicalGluing}
Let $A,B$ be both closed (resp. both open) subsets of a topological space $X$ such that $X = A \cup B$, and let $Y$ also be a topological space. If $f :A\to Y$ and $g : B \to Y$ are continuous and if $f \cup g : X \to Y$ is a map, then $f \cup g$ is continuous.
\end{lemma}

In Section \ref{sec:proof}, we apply this result in combination with the next lemma (in the specific case where $X$ is a union of some CAD cells),  whose proof relies on elementary topological arguments. 

\begin{lemma}\label{lemma:consecutive-closed}
        Let $S$ be a semi-algebraic set of $\R^{n-1}$ and consider the sectors and the section\footnote{Recall that this implies that $f_1 \in \mathcal{S}^0(S) \cup \{-\infty\}, f_2 \in \mathcal{S}^0(S), f_3 \in \mathcal{S}^0(S) \cup \{+\infty\}$ and that $f_1 < f_2 < f_3$ on $S$ (see Subsection \ref{sub:CAD}).}
        $C = S \odot (f_1,f_2), D = S \odot \{f_2\}, E = S \odot (f_2,f_3).$
        In $C \cup D \cup E$, the sets $C \cup D$ and $D \cup E$ are closed and the sets $C$ and $E$ are open.
\end{lemma}
% \begin{proof}
%    Since $C,D$ and $E$ are pairwise disjoint, $E$ is the relative complement of $C \cup D$ with respect to $C \cup D \cup E$. To show that $C \cup D$ is closed, it is equivalent to show that $E$ is open. We consider the two continuous functions 
%     $$F_i : C \cup D \cup E \to \R : (\textbf{x}, x_n) \mapsto x_n - f_i(\textbf{x}) \quad \text{ for } i \in \{2,3\},$$
%     and directly obtain that $E = F_2^{-1}\left((0,\infty)\right) \cap F_3^{-1}\left((-\infty,0)\right)$, which is the intersection of two open sets. 
%     Similar arguments hold for $D \cup E$ and $C$.
% \end{proof}

At one stage of the argument in Section \ref{sec:proof}, it will be convenient to have a version of the pasting lemma formulated in terms of limits and some adjacent CAD cells of the plane.

\begin{lemma}\label{prop:gluingLimit}
    Let $C, D\subseteq \R^2$ be two CAD cells such that $D \subseteq \overline{C} \setminus C$, and let $f : C \to \R$ and $g : D \to \R$ be two continuous functions.
    If for every $p \in D$, we have 
    \begin{equation}\label{eqn:lim}
    \lim_{\substack{x \to p\\x\in C}}f(x)=g(p),
    \end{equation}
    then the function $f \cup g : C \cup D \to \R$ is continuous.
\end{lemma}
\begin{proof}
    For every $p \in C \cup D$, we show that $f \cup g$ is continuous at $p$. If $p \in C$, this follows from the continuity of $f$ at $p$, together with the existence of a neighbourhood $N \subset \R^2$ of $p$ that does not intersect $D$ (which follows from a simple dimension argument on~$C$). If $p \in D$, the continuity follows from that of $g$ at $p$ and from~\eqref{eqn:lim}.
\end{proof}

\subsection{Every CAD cell of the plane is locally boundary connected}\label{sec:lbc}

% \todo[inline]{Refaire toute cette transition. Dire que le goal ici est le Corollaire \ref{cor:bypass}, que les résultats sont intéressants en tant que tels}
% \todo[inline]{Dire qu'on prend un peu de hauteur dans le but de montrer qu'à un plongement près dans $[-\infty, \infty]$, les sections de CAD sont équirégulières à $\{0\}^2, (0,1)^2, ...$, et on en déduit le thrm 3.22.}

The main goal of this section is Corollary~\ref{cor:bypass}, the final ingredient needed for the proof of the main theorem, and which based on the closedness and curtainedness assumption. 

To this end, we establish several topological results of independent interest in CAD theory. In particular, we show that every CAD cell in the plane is locally boundary connected (Theorem~\ref{thrm:CADcellsPlanelbc}), a result that appears to be missing from the literature. This is achieved by observing that such cells are, up to an embedding into $[-\infty, \infty]^2$, equiregular to one of the sets $\{0\}^2, (-1,1)\times \{0\}$ or $(-1,1)^2$ (Proposition \ref{prop:CADcellsLowdim}), and that this classification suffices for our purposes (Proposition \ref{prop:lbp-embedding}). As a corollary, we will obtain that for every CAD cell $C$ of $\R^3$, and every $p \in \R^2$, the fibre $\overline{C} \cap \left(\{p\} \times \R\right)$ is always a closed segment. 
We then apply this result under the closedness and curtainedness assumptions to obtain Corollary~\ref{cor:bypass}. 

%We recall some definitions of \cite{DLSregular} in order to prove that every CAD cell of $\R$ and of $\R^2$ is locally boundary connected (defined below). As a corollary, we will obtain that for every CAD cell $C$ of $\R^3$, and every $p \in \R^2$, the fibre $\overline{C} \cap \left(\{p\} \times \R\right)$ is always a closed segment. This plays a central role in the proof of the main theorem.

\subsubsection{Locally boundary $\mathscr{P}$}

As in \cite{DLSregular}, since extending the arguments from the locally boundary connected setting to the more general locally boundary $\mathscr{P}$ setting (where $\mathscr{P}$ is any topological property) requires no additional effort, we state and prove the results directly in this broader framework.
A topological property $\mathscr{P}$ is a property of topological spaces that is invariant under homeomorphisms, i.e. for every pairs of homeomorphic topological spaces $X$ and $Y$, if $X$ satisfies $\mathscr{P}$, then $Y$ satisfies $\mathscr{P}$.
For instance, the connectedness is a topological property.

\begin{definition}[See \cite{DLSregular}]\label{def:locallyboundaryP}
    Let $\mathscr{P}$ be a topological property and let $X$ be a subset of a topological space $X'$. 
    The (cell) boundary  of $X$ (in $X'$) is $\partial X = \overline{X} \setminus X$.
    We say that $X$ is locally boundary $\mathscr{P}$ (in $X'$)  if every $p \in \partial X$ has a base of neighbourhoods $\mathcal{N}$ in $X'$ such that for every $N \in \mathcal N$, the set $X \cap N$ satisfies $\mathscr{P}$.  
\end{definition}

An embedding is an injective continuous map $e : X' \to Y'$ between two topological spaces which is a homeomorphism between $X'$ and $e(X')$. We provide a criterion useful to show that some subset $X \subseteq X'$ is locally boundary $\mathscr{P}$ in $X'$.
%For every subset $X \subseteq X'$, we show that if the embedded set $e(X)$ is locally boundary $\mathscr{P}$, then $X$ is locally boundary $\mathscr{P}$ as well.

\begin{proposition}\label{prop:lbp-embedding}
    Let $\mathscr{P}$ be a topological property and $e : X' \to Y'$ be an embedding. For every $X \subseteq X'$, if $e(X)$ is locally boundary $\mathscr{P}$ in $Y'$, then $X$ is locally boundary $\mathscr{P}$ in $X'$.
\end{proposition}
\begin{proof}
    For $X \subseteq X'$, we write $Y = e(X)$. We consider $p \in \partial X$ and construct a base of neighbourhoods $\mathcal{N}$ in $X'$ as in Definition \ref{def:locallyboundaryP}.
    Since $e$ is continuous, it is clear that $e(\overline{X}) \subseteq \overline{e(X)} = \overline{Y}$. Furthermore, the injectivity of $e$ implies $e(\partial X) \subseteq  \partial Y$.
   % $$e(\partial X) = e (\overline{X} \setminus X) = e(\overline{X}) \setminus e(X) \subseteq \overline{Y} \setminus Y = \partial Y.$$ 
   This shows that $e(p) \in \partial Y$. By assumptions, there exists a base of neighbourhoods $\mathcal{M}$ in $Y'$ such that for every $M \in \mathcal{M}$, the set $Y \cap M$ satisfies $\mathscr{P}$. We set $\mathcal{N} = e^{-1}(\mathcal{M})$ and obtain the conclusion. %(For every $N \in \mathcal{N}$, there exists $M \in \mathcal{M}$ such that $N = e^{-1}(M)$. In particular, $N$ is a neighbourhood of $p$ in $X'$ since $e$ is continuous. Moreover, the set $X \cap N = e^{-1}(Y) \cap e^{-1}(M) = e^{-1}(Y \cap M)$ satisfies $\mathscr{P}$ since it is a topological property, and since $e$ is a homeomorphism between $X$ and $Y$.)
\end{proof}

% \lucas[inline]{Inutile. Retirer cette partie jusque 3.5.2 ?}
% \lucas[inline]{on a besoin de la def de homeo de pair plus loin}
We recall that given the inclusions of topological spaces $X \subseteq X'$ and $Y \subseteq Y'$, a homeomorphism (of pairs) $\varphi : (X', X) \to (Y ', Y )$ is a homeomorphism $\varphi : X' \to Y '$ such that $\varphi(X) = Y$. A direct consequence of the previous result is that locally boundary $\mathscr{P}$ is preserved by homeomorphism of pairs. 
% \begin{proposition}
%      Let  $\varphi : (X', X) \to (Y ', Y )$ be a homeomorphism of pairs and $\mathscr{P}$ be a topological property. If $X$ is locally boundary $\mathscr{P}$ (in $X'$), then $Y$ is locally boundary $\mathscr{P}$ (in $Y'$),
% \end{proposition}
% \begin{proof}
%         This is a direct application of Proposition \ref{prop:lbp-embedding}, together with the fact that $\varphi^{-1} : Y' \to X'$ is an embedding. %sending $Y$ to~$X$.
%     % We consider $p \in \partial Y$ and we obtain a base of neighbourhood $\mathcal{N}$ as in Definition \ref{def:locallyboundaryP}.
%     % By assumption on $\varphi$, it is clear that $\varphi(\partial X) = \partial Y$. Thus, there exists $q \in \partial X$ such that $p = \varphi (q).$ 
%     % Since $X$ is locally boundary $\mathscr{P}$, there exists a base of neighbourhoods $\mathcal{M}$ in $X'$ such that for every $M \in \mathcal{M}$, the set $X \cap M$ satisfies $\mathscr{P}$. We set $\mathcal{N} = \varphi(\mathcal{M})$ and the conclusion follows. %(For every $N \in \mathcal{N}$, there exists $M \in \mathcal{M}$ such that $N = \varphi(M)$. The set $Y \cap N = \varphi(X) \cap \varphi(M) = \varphi(X \cap M)$ satisfies $\mathscr{P}$ since it is a topological property.)
% \end{proof}

\subsubsection{CAD cells and the extended real line}

We consider the extended real line $[-\infty, \infty] = \R \cup \{-\infty, \infty\}$. This set is totally ordered by the extension of the usual order on $\R$ given by $-\infty \leqslant x \leqslant \infty$ for every $x \in [-\infty, \infty]$. The order topology on $[-\infty, \infty]$ makes it a two-point compactification of $\R$. More precisely, the inclusion
$$i : \R \to [-\infty, \infty] : x \mapsto x$$ is an embedding such that $i(\R)$ is dense in the compact $[-\infty, \infty]$, and $[-\infty, \infty] \setminus i(\R) = \{-\infty, \infty\}$ contains two points.

Following \cite{DLSregular}, given the inclusions of topological spaces $X \subseteq X'$ and $Y \subseteq Y'$, we say that $X$ and $Y$ are equiregular if there exists a homeomorphism $\varphi : (\overline{X}, X) \to (\overline{Y}, Y )$. In particular, $(-\infty, \infty) \subseteq [-\infty, \infty]$ and $(0,1) \subseteq \R$ are equiregular via the homeomorphism  $\mathfrak{B} : ([-\infty, \infty], (-\infty, \infty)) \to ([0,1], (0,1))$ defined by 
$$\mathfrak{B}(y) = \begin{cases}
    0 & \text { if } y =-\infty,\\
    \frac{1}{2}(1+\frac{y}{|y| + 1}) & \text { if } y \in (-\infty, \infty),\\
    1 & \text { if } y =\infty.
\end{cases}$$

We show that the embedding of CAD cells of $\R$ (resp. of $\R^2$)  in $[-\infty, \infty]$ (resp. $[-\infty, \infty]^2$) are simple equiregular objects.

\begin{lemma}\label{lemma:equiregularIntervals}
    Let $a',b' \in \R$ such that $a' < b'$. If $l', u' : [a',b'] \to \R$ are continuous and satisfy $l' < u'$ on $(a',b')$, 
    then $(a',b') \odot (l',u')$ is equiregular to $(0,1)^2$.
\end{lemma}
\begin{proof}
 A direct computation shows that
 $$\overline{(a',b')\odot(l',u')} = [a',b'] \odot[l',u'],$$
 where the right-hand side is an obvious adaptation of the $\odot$ notation.
 If $l'(a') < u'(a')$ and if $l'(b') < u'(b')$, then we consider the two auxiliary maps $f : \R^2 \times [0,1] \to \R$ and $g : [0,1]^2 \to \R$ given by
 \begin{align*}
        f(x,y,t) &= (1-t)x + ty\\
        g(t,s) &= f\left(l(f(a',b',t)), u(f(a',b',t), s)\right) 
    \end{align*}
    for every $x,y \in \R, t \in [-1,1]$. Note that $t \mapsto f(x,y,t)$ parametrizes the segment $[x,y]$.
    The conclusion follows since
    \begin{align*}
        ([0,1],(0,1)) &\to ([a',b'] \odot[l',u'], (a',b')\odot(l',u')) \\
        (t,s) &\mapsto \left( f(a',b',t),g(t,s) \right)
    \end{align*}
    is a homeomorphism of pairs. 
    The proofs for the other cases (where $l'(a') = u'(a')$ but $l'(b') < u'(b')$, $l'(a') < u'(a')$ but $l'(b') = u'(b')$ or $l'(a') = u'(a')$ and $l'(b') = u'(b')$) are analogous. %In each of these cases, we show that $(a',b')\odot(l',u')$ is respectively equiregular either to a triangular region $(0,1) \odot (0, \text{id}_{\R}), (0,1) \odot (0,1- \text{id}_{\R})$, or to the unit disk. All these sets are themselves equiregular to $(0,1)^2$.
\end{proof}
\begin{proposition}\label{prop:CADcellsLowdim}\mbox{}
    %If $C$ is a CAD cell of $\R$, then $i(C)$ is equiregular to $\{0\}$ or to $(-1,1)$.
    If $C$ is a CAD cell of $\R^2$, then $i\times i(C)$ is equiregular to $\{0\}^2, (-1,1)\times \{0\}$ or to $(-1,1)^2$.
\end{proposition}
\begin{remark}
    This result does not generalizes naively in $\R^n$ with $n \geqslant 3$. For instance, the Cornet cell $\mathfrak{C}$ (see Definition 3.2 of \cite{exotic}) is a CAD cell of $\R^3$ such that $i \times i\times i(\mathfrak{C})$ is not equiregular to $\{0\}^3, (-1,1)\times \{0\}^2,(-1,1)^2\times \{0\}$ or to  $(-1,1)^3$. Indeed, $i(\mathfrak{C})$ is not locally boundary connected, but the four others are. 
\end{remark}
\begin{proof}
        %If $C = \{\xi\}\odot \{\xi'\}$ (section-section) is direct. 
        %If $C = (a,b) \odot (l,u)$ (sector-sector), then we denote $D = (a,b)$ and by $\varphi$ a homeomorphism of pairs between $(\overline{i(D)}, i(D))$ and $([-1,1],(-1,1))$ (see the previous item). The set $i \times i (C)$
        If $C = (a,b) \odot (l,u)$ (sector-sector type), %then if we write $D = (a,b)$, then $i \times i (C) = i(D) \odot (\widetilde{l}, \widetilde{u})$ where $\widetilde{l} \circ i = l$ and $\widetilde{u} \circ i = u$ on $D$. Moreover, we have
        then $i \times i(C)$ is equiregular to 
        $\mathfrak{B}\times \mathfrak{B}(i \times i(C)) = (a',b') \odot (l',u')$ where $a', b' \in \R, a' < b'$ and $l',u'\in \mathcal{S}^0(a',b')$ are bounded and satisfy $l' < u'$ on $(a',b')$. By Proposition 2.5.3 of \cite{bochnaketal1998}, the continuous semi-algebraic maps $l',u'$ both admit a continuous extension on $[a',b']$, that we still denote by $l'$ and $u'$. By Lemma \ref{lemma:equiregularIntervals}, $\mathfrak{B} \times \mathfrak{B}(i \times i(C))$ is equiregular to $(0,1)^2$, and the conclusion follows. 
        
        The proof for the other type of CAD cells of $\R^2$ %(section-sector, sector-section, section-section)
        is analogous. 
\end{proof}

\begin{theorem}\label{thrm:CADcellsPlanelbc}
    Every CAD cell of $\R^2$ is locally boundary connected. 
\end{theorem}
\begin{proof}
    This follows from Propositions \ref{prop:lbp-embedding} and \ref{prop:CADcellsLowdim}. %For instance, if $C$ is a CAD cell of $\R^2$, then the embedded $i\times i(C)$ is locally boundary connected, since it is equiregular to $\{0\}^2, (-1,1)\times \{0\}$ or to $(-1,1)^2$.
\end{proof}

\subsubsection{Applications of the locally boundary connectedness}
We provide an independent, elementary and short proof of a generalization of Theorem 3.3.31 of \cite{arnon} and Proposition 5.2 of \cite{lazard2010}, both in terms of dimension and in the assumptions involved. In doing so, we highlight the precise role played by the locally boundary connectedness property.
For the remainder of this subsection, if $S \subseteq \R^{n+1}$, then $$S' = (\id_{\R^n} \times i)(S).$$ This is the set $S$, but seen as a subset of $\R^n \times [-\infty, \infty]$.

\begin{proposition}\label{prop:fibrelbc}
    Let $C$ be a CAD cell of $\R^{n+1}$. If $\pi_n(C)$ is locally boundary connected, then for every $p \in \partial \pi_n(C)$, the fibre $\overline{C'} \cap \left(\{p\}\times [-\infty, \infty]\right)$ is a closed segment.
\end{proposition}
\begin{proof}
    We consider $(p,y_1), (p,y_2) \in \overline{C'}$ with $y_1 < y_2$ and we show that for every $y \in (y_1,y_2)$, we have $(p,y) \in \overline{C'}$, i.e. for every neighbourhood $U$ of $(p,y)$ in $\R^n \times [-\infty, \infty]$, the intersection $U \cap C'$ is not empty. 
    
    Since the projection $\pi_n$ is open, $\pi_n(U)$ is a neighbourhood of $p$ in $\R^n$. By assumption on $\pi_n(C)$, there exists a neighbourhood $M \times N \subseteq U$ of $(p,y)$ in $\R^n \times [-\infty, \infty]$ such that the intersection $M \cap \pi_n(C)$ is connected. 
    It is straightforward to construct two neighbourhoods $V_1$ and $V_2$ of $y_1$ and $y_2$ in $[-\infty, \infty]$ such that for every $v_1 \in V_1, v_2 \in V_2$, $v_1 < y < v_2$. For $i \in \{1,2\}$, since the set $M \times V_i$ is a neighbourhood of $(p,y_i)$ in $\R^{n} \times [-\infty, \infty]$, there exists $(q_i,z_i) \in C' \cap ( M \times V_i)$. We consider the projection %to the last component
    $$f : C' \cap (M \times [-\infty, \infty]) \to [-\infty, \infty] : (a,b) \mapsto b.$$
    Since $C' \cap (M \times [-\infty, \infty])$ is connected, $f$ is continuous, and since $f(q_1,z_1) = z_1 < y < z_2 = f(q_2,z_2)$, we obtain $(q, y) \in C' \cap (M \times [-\infty, \infty])$. The conclusion follows since $(q,y) \in C' \cap (M \times N)$, and the latter is a subset of $C' \cap U$ by construction. 
\end{proof}

\begin{corollary}\label{cor:fibreSpace}
    Let $C$ be a CAD cell of $\R^{3}$. If $p \in \partial \pi_2(C)$, then the fibre $\overline{C'} \cap \left(\{p\}\times [-\infty, \infty]\right)$ is a closed segment.
\end{corollary}
\begin{proof}
   This follows from Theorem \ref{thrm:CADcellsPlanelbc}, Proposition \ref{prop:fibrelbc}.
\end{proof}

Finally, we apply the results of this section to a specific case arising naturally in Section~\ref{sec:proof} (where we make use of Lemma~\ref{prop:gluingLimit}).

\begin{corollary}\label{cor:bypass}
    Let $S$ be a closed and curtained semi-algebraic set of $\R^3$, $D = C \odot \{f\}$ be a CAD cell of $\R^3$ such that $D \subseteq S$ and $p \in \partial C$. 
    If $S$ has no curtain at $p$, then there exists a unique $b \in [-\infty,\infty]$ such that $\overline{D'} \cap (\{p\}\times [-\infty, \infty]) = \{(p,b)\}.$
    Moreover, we have 
    \begin{equation}\label{eqn:limf=b}
        \lim_{\substack{x \to p\\x\in C}}f(x) = b,
    \end{equation}
    and if $b \notin \{-\infty, \infty\}$, then $(p,b) \in S$.
\end{corollary}

Recall that Equation \eqref{eqn:limf=b} means that for every neighbourhood $V$ of $b$ in $[-\infty, \infty]$, there exists a neighbourhood $U$ of $p$ in $\R^n$ such that $f(C \cap U) \subseteq V$.

\begin{proof}
    The uniqueness is direct by Corollary \ref{cor:fibreSpace} and the assumptions on $S$. 
    For the existence, we consider $(p_k)_{k \in \N}$ a sequence in $C$ converging to $p$. Hence, $\left(\mathfrak{B}(i(f(p_k)))\right)_{k \in \N}$ is a bounded sequence in $[0,1]$. By the Bolzano–Weierstrass theorem, we can extract a converging subsequence  $\left(\mathfrak{B}(i(f(p_{m(k)})))\right)_{k \in \N}$ where $m : \N \to \N$ is strictly increasing. Hence, the sequence $$\Big(\left(p_{m(k)}, i(f(p_{m(k)}))\right)\Big)_{k \in \N}$$ converges to a point $(p, b) \in \R^2 \times [-\infty, \infty]$, so $\overline{D'} \cap (\{p\}\times [-\infty, \infty])$ is not empty. Equation \eqref{eqn:limf=b} is direct by Lemma 3.3.12 of \cite{arnon}. Finally, if $b \notin \{-\infty, \infty\}$, then $(p,b) \in \overline{D}$, and the latter is a subset of $S$ since $S$ is closed. 
\end{proof}

\section{Proof of the main theorem}\label{sec:proof}

Theorem \ref{thrm:main} is a direct consequence of the following stronger result. Indeed, every algebraic set is closed, curtained and semi-algebraic (see Lemma \ref{prop:curtained-algebraic}).
\begin{theorem}\label{thrm:main-curtained}
    Every finite family $\Fc = (S_1, \ldots, S_p)$ of closed and curtained semi-algebraic sets of $\R^3$ admits a minimum CAD.
\end{theorem}
%\lucas[inline]{Note that all the counter-examples of \cite{miniCAD-extended} are either not closed, or closed but not curtained (see Definition 4.3 and Example 4.9 of \cite{miniCAD-extended}).}
\begin{proof}
%Let $\Fc = (S_1, \ldots, S_p)$ be a finite family of closed and curtained semi-algebraic sets of $\R^3$.
We show that the poset $(\CAD(\Fc),\preceq)$ admits a minimum element. By Proposition \ref{prop:conflRed}, it is sufficient to show that the reduction system $(\CAD_{3-\text{irr}}(\Fc), \leftarrow)$ is confluent.

Let $\Cr, \Dr, \Er \in \CAD_{3-\text{irr}}(\Fc)$ such that there exist reduction rules $\Phi_A$ and $\Phi_B$ (with $A \neq B$) from $\Cr$ to $\Dr$ and $\Er$ respectively. We need to show that there exists $\Fr \in \CAD_{3-\text{irr}}(\Fc)$ such that $\Dr \stackrel{*}{\rightarrow} \Fr \stackrel{*}{\leftarrow} \Er$. 
% \lucas{useless?}The situation is depicted in Figure \ref{fig:ideaProof}.
% \begin{figure}[H]
%     \centering
%     \[\begin{tikzcd}
% 	& \Cr \\
% 	\Dr && \Er \\
% 	& \Fr
% 	\arrow["{\Phi_A}"', from=1-2, to=2-1]
% 	\arrow["{\Phi_B}", from=1-2, to=2-3]
% 	\arrow["{*}"', dashed, from=2-1, to=3-2]
% 	\arrow["{*}", dashed, from=2-3, to=3-2]
% \end{tikzcd}\]
%     \caption{Idea of the proof}
%     \label{fig:ideaProof}
% \end{figure}
Note that Lemma \ref{lemma:thrm59Red} asserts that it is equivalent to show the existence of $\Fr \in \CAD_{3-\text{irr}}(\Fc)$ such that $\Dr \succeq \Fr \preceq \Er$.

%An immediate case is when $A = B$. In this situation, we directly have $\Dr = \Er$ and we may conclude by taking $\Fr = \Dr$. 
%For the remaining case, when $A \neq B$, we prove the following result. 
%\lucas{le cas $A = B$ étant trivial, on considère $A \neq B$.}
%\begin{claim}\label{prop:phipsiAB}
    %Suppose that $A \neq B$. 
We now show that up to a transposition of $(A, \Dr)$ and $(B, \Er)$, the CAD reduction rule $\Phi_{\psi_A(B)}$ is defined on $\Dr$, and the corresponding reduced CAD $\Fr$ is smaller than $\Er$.
%\end{claim}
The proof of Theorem \ref{thrm:main-curtained} is finished once this claim is established. The remainder of this section is dedicated to proving it.
%\subsection{Proof of Claim \ref{prop:phipsiAB}}

To show that $\Phi_{\psi_A(B)}$ is defined on $\Dr$, we proceed by a proof by cases on the length of $A$ and $B$ in Subsections \ref{case:|A|=|B|} and \ref{case:|A|neq|B|}. Since the CADs considered here are all $3$-reduced, $A$ and $B$ must have length $1$ or $2$. In Subsection \ref{subsection:FpreceqC}, we show that $\Fr \preceq \Er$.

We write $\xi_I$ (resp. $\delta_I, \varepsilon_I, \varphi_I$) the functions used to define the decompositions $\Cr$ (resp. $\Dr, \Er, \Fr$) as in Definition \ref{def:cad}, where $I$ is an even index of $\N, \N^2$ or $\N^3$.

In both cases below, after the potential transposition of $(A, \Dr)$ and $(B, \Er)$, we show that $\Psi_{\psi_A(B)}$ is always defined on $\Tree(\Dr, \Fc)$, so that it suffices to check if this tree reduction lifts to a CAD reduction using Proposition \ref{prop:liftingReduction}.

\subsection{Case $A \neq B, |A| = |B|$}\label{case:|A|=|B|}
Up to a transposition, we can suppose that $A >_\text{lex} B$ (where $>_\text{lex}$ denotes the usual lexicographic ordering). In this case, $B \in F_A$ and hence $\psi_A(B) = B$. We write $k = |A| = |B| \in \{1,2\}$. 

We now show that the CAD tree reduction $\Psi_{\psi_A(B)} (= \Psi_B)$ is defined on 
$\Tree(\Dr, \Fc) = (T',L')$.
By definition, this holds if and only if
\begin{equation}\label{eqn:L'(B)}
    L'(B-e_k) = L'(B) = L'(B+e_k).
\end{equation}
For any tuple $I$ having prefix $B-e_k$, $B$, or $B+e_k$, we have $I <_\text{lex} A$ (where $A$ is padded with trailing zeros when $|I| > |A|$).  
Hence $I \in F_A$ and $\psi_A(I) = I$.  
Since $T' = \psi_A(T)$, the subtrees of $T'$ with prefixes $B-e_k$, $B$, and $B+e_k$ coincide with the corresponding subtrees of $T$, where $\Tree(\Cr, \Fc) = (T,L)$. 
If $I$ is in addition a leaf of $T'$, then by definition of the reduction $\Psi_A$ (defined from $\Tree(\Cr, \Fc)$ to $\Tree(\Dr, \Fc)$) we have $L'(I) = L'(\psi_A(I)) = L(I).$
Thus, Equation~\eqref{eqn:L'(B)} simply becomes $L(B-e_k) = L(B) = L(B+e_k),$
which is indeed satisfied since $\Psi_B$ is defined on $\Tree(\Cr, \Fc)$.

We now check that the CAD tree reduction $\Psi_{B}$ lifts to a CAD reduction defined on $\Dr$.
% If $k = 3$, then we write $B = b_1b_2b_3 \in (\N^*)^3$  and observe that $V$ has a curtain at $D_{\psi_A(b_1b_2)} = C_{b_1b_2}$ by Proposition \ref{prop:curtain}. Using the same result again, we obtain that $\Phi_{\psi_A(B)}$ is defined on $\Dr$, as requested.
%We divide the discussion of the two cases ($k = 1$ and $k = 2$) into two subcases. \todo[inline]{
We deal with both cases $k = 1$ and $k = 2$ at once, but divide the discussion in two cases according to whether the condition $A = B + 2e_{k}$ holds or not.
Intuitively, in the first case, $B$ and $A$ denote successive sections of the CAD $\Cr_k$ and we merge five cells of $\Cr_{k}$ in order to obtain $\Fr_{k}$, whereas in the second case, we merge separately two groups of three cells each. Then, the levels above the $k$-th one of the CAD $\Fr$ are obtained by merging the respective cells above those already merged at level $k$.
 
%The first subcase occurs when $B$ and $A$ denote successive sections of the CAD $\Cr_k$ (i.e.\ $A = B + 2e_{k}$), and the second subcase is the remaining one (i.e.\ $A >_\text{lex} B + 2e_{k}$). 

%Intuitively, in the first subcase\lucas{subcase $\to$ case}, we merge five cells of $\Cr_{k}$ in order to obtain $\Fr_{k}$, whereas in the second subcase, we merge separately two groups of three cells each. 
%The levels above the $k$-th one of the CAD $\Fr$ are obtained by merging the respective cells above those already merged at level $k$.

By Proposition \ref{prop:liftingReduction}, to show that the CAD reduction rule $\Phi_B$ is defined on $\Dr$, it is sufficient to show that for every even $I' \in S_{B-e_k} \cap \psi_{B}(T')$, the function 
$\varphi_{I'} = \delta_{I'} \cup \delta_{I'+e_k} \cup \delta_{I'+2e_k}$
is continuous. By definition, if $I$ is an even index such that $D_I \in \Dr_2 \cup \Dr_3$, then
    \begin{align*}
        \delta_{I} = 
        \begin{cases}
            \xi_{I} \cup \xi_{I+ e_k} \cup \xi_{I + 2 e_k} &\text{ if } I \in S_{A - e_k}, \\
            \xi_{I + 2 e_k} &\text{ if } I \in S_A \cup N_A, \\
            \xi_{I} &\text{ otherwise.}
        \end{cases}
    \end{align*}
    
We suppose that $A = B + 2e_{k}$. In this case, $A-e_k = B + e_k$ and since $I' \in S_{B-e_k}$, it is clear that $I', I'+e_k \notin S_{A-e_k} \cup S_A \cup N_A$, but $I'+2e_k \in S_{A-e_k}$.  In particular, we have 
\begin{align*}
    \varphi_{I'} =& \xi_{I'} \cup \xi_{I'+e_k}\cup \xi_{I'+2e_k}\cup \xi_{I'+3e_k}\cup \xi_{I'+4e_k}\\
    = &\left(\xi_{I'} \cup \xi_{I'+e_k}\cup \xi_{I'+2e_k} \right)\\
     & \phantom{eeeeeeeeee}\cup\left(\xi_{I'+2e_k}\cup \xi_{I'+3e_k}\cup \xi_{I'+4e_k}\right),
\end{align*}
where the last equality is obtained by a rearrangement of terms.
Since $\varepsilon_{I'} = \xi_{I'} \cup \xi_{I'+e_k}\cup \xi_{I'+2e_k}$ and $\delta_{I'+2e_k} = \xi_{I'+2e_k}\cup \xi_{I'+3e_k}\cup \xi_{I'+4e_k}$, we obtain that $\varphi_{I'}$ is the union of these two functions, whose domains are open in the domain of $\varphi_{I'}$ by Lemma \ref{lemma:consecutive-closed}. The map $\varphi_{I'}$ is thus continuous by Lemma \ref{lemma:classicalGluing}.

We now suppose that  $A >_\text{lex} B + 2e_{k}$. In this case, $I', I'+e_k, I'+2e_k \notin S_{A-e_k} \cup S_A \cup N_A$ and thus 
$\varphi_{I'} = \xi_{I'} \cup \xi_{I'+e_k} \cup \xi_{I'+2e_k}.$
This union is precisely the definition of $\varepsilon_{I'}$, which is continuous by assumption.

%This concludes the proof of Claim \ref{prop:phipsiAB} in the case where $A \neq B$ and $|A| = |B|$. \lucas{remove with the removing of the claim 1}

\subsection{Case $A \neq B, |A| \neq |B|$}\label{case:|A|neq|B|}
Up to a transposition, we can suppose that $1 = |A|<|B| = 2$, and we write $A = a_1 \in \N^*, B = b_1b_2 \in (\N^*)^2$.

A discussion analogous to the one of the preceding case shows that 
$\Psi_{\psi_A(B)}$ is indeed defined on $\Tree(\Dr, \Fc) = (T',L')$.

We now show that this CAD tree reductions lifts to a CAD reduction $\Phi_{\psi_A(B)}$ defined on $\Dr$.
%If $|B| = 3$, then $B = b_1b_2b_3 \in (\N^*)^3$ and $V$ has a curtain at $C_{b_1b_2} \in \Cr_2$ by Proposition \ref{prop:curtain}. Using again this result, it is sufficient to show that $V$ has a curtain at $D_{\psi_A(b_1b_2)} \in \Dr_2$ to conclude the proof in this case. This is straightforward since $C_{b_1b_2} \subseteq D_{\psi_A(b_1b_2)}$ by definition of $\Dr$.
By Proposition \ref{prop:liftingReduction}, it is sufficient to show that for every even $I' \in S_{\psi_A(B)-e_2} \cap \psi_{\psi_A(B)}(T')$, the function 
    $\varphi_{I'} = \delta_{I'} \cup \delta_{I'+e_2} \cup \delta_{I'+2e_2}$
is continuous. We assume that there exists at least one such index $I'$ (otherwise the proof is over).

If $B \notin S_{A-e_1} \cup S_A \cup S_{A+e_1}$, we proceed as above to compute
$$\varphi_{I'} = \begin{cases}
    \varepsilon_{I'+2e_1} & \text{ if } B \in N_A,\\
    \varepsilon_{I'} & \text{ otherwise.}
\end{cases}$$
In both cases, the assumptions on $\Er$ imply that $\varphi_{I'}$ is continuous.  

Now, we assume that $B \in S_{A - e_1} \cup S_A \cup S_{A + e_1}$. We will treat the case $B \in S_{A - e_1}$. The other two cases ($B \in S_A$ and $B \in S_{A + e_1}$) are handled in the same way.
We have $B = (a_1-1,b_2)$ and $\psi_A(B) = B$. In particular, we obtain $I' = (a_1-1,b_2-1,2j)$ for some $j \in \N^*$.

%\begin{claim}\label{claim:technicalContinuity} 
We now show that the functions $\delta_{I'} \cup \delta_{I'+e_2}$ and $\delta_{I'+e_2} \cup \delta_{I'+2e_2}$ are continuous on $D_{(a_1-1, b_2-1)} \cup D_{(a_1-1, b_2)}$ and $D_{(a_1-1, b_2)} \cup D_{(a_1-1, b_2+1)}$ respectively.
The fact that $\varphi_{I'}$ is continuous follows directly (in combination with Lemma \ref{lemma:classicalGluing} and Lemma \ref{lemma:consecutive-closed}).
%The proof of this claim is delayed at the end of the section.
% The conclusion of the proof of Claim \ref{prop:phipsiAB} in the case $|A|=1, |B|=2$ follows from Claim \ref{claim:technicalContinuity}. Indeed, combined with Lemma \ref{lemma:classicalGluing} and Lemma \ref{lemma:consecutive-closed}, we obtain that $\varphi_{I'}$ is continuous.
%Similar developments show the continuity of $\delta_{I'+e_2} \cup \delta_{I'+2e_2}$.
%\begin{proof}[Proof of Claim \ref{claim:technicalContinuity}.]

To show that $\delta_{I'} \cup \delta_{I'+e_2}$ is continuous on $D_{(a_1-1, b_2-1)} \cup D_{(a_1-1, b_2)}$, it is equivalent (see Lemma \ref{prop:gluingLimit}) to show that for every $(x_0, y_0) \in D_{(a_1-1,b_2)}$, we have 
$$\lim_{\substack{(x,y) \to (x_0,y_0)\\(x,y)\in D_{(a_1-1, b_2-1)}}}\delta_{I'}(x,y)=\delta_{(a_1-1,b_2,2j)}(x_0,y_0).$$ %\todo{$I'+e_2 \to (a_1-1,b_2,2j)$}
%As it will play a crucial role below, recall that $D_{(a_1-1,b_2)}$ is connected since it is the section $D_{a_1-1} \odot \delta_{a_1-1,b_2}$, where $D_{a_1-1}$ is an open interval and $\delta_{a_1-1,b_2}$ is continuous. 
    %Recall that $I'+e_2 = (a_1-1,b_2,2j)$.
    We denote by $L \in \N$ the number of sections of $\Dr$ with base $D_{(a_1-1,b_2)}$. Since $\Psi_{\psi_A(B)} = \Psi_{B}$ is defined on $\Tree(\Dr, \Fc)$, the CAD $\Dr$ has also $L$ sections with base $D_{(a_1-1,b_2-1)}$. Hence, $L \neq 0$. 
    By convention (see Definition \ref{def:cad}), $\delta_{a_1-1,b_2,0} = -\infty$ and $\delta_{a_1-1,b_2,2(L+1)} = \infty$. For every $l \in \{0, 1, \ldots, L, L+1\}$, we define $\Omega_{l}$ to be the set of those points $(x_0,y_0) \in D_{(a_1-1,b_2)}$ satisfying
    $$\lim_{\substack{(x,y) \to (x_0,y_0)\\(x,y)\in D_{(a_1-1, b_2-1)}}}\delta_{I'}(x,y)=\delta_{(a_1-1,b_2,2l)}(x_0,y_0).$$
    We then have to prove that $D_{(a_1-1,b_2)} = \Omega_j$. It follows from the connectedness of $D_{(a_1-1,b_2)}$ together with the following assertions, which we prove successively:
    \begin{enumerate}[(i)]
        \item We have $D_{(a_1-1,b_2)} = \Omega_0 \cup \Omega_1 \cup \ldots \cup \Omega_L \cup \Omega_{L+1}$. \label{item:OmegaCovering}
        \item The sets $\Omega_{0}, \ldots, \Omega_{L+1}$ are pairwise disjoint. \label{item:disjoint}
        \item The sets $\Omega_{0}, \ldots, \Omega_{L+1}$ are open in $D_{(a_1-1,b_2)}$. \label{item:open}
        \item The set $\Omega_j$ is not empty. \label{item:notempty}
    \end{enumerate}    
   We first observe that $\Omega_l \subseteq D_{(a_1-1,b_2)}$ for every $l \in \{0, \ldots, L+1\}$ by definition. For the other inclusion of Assertion \ref{item:OmegaCovering}, let $(x_0,y_0) \in D_{(a_1-1,b_2)}$. We show that there exists $l\in \{0, \ldots, L+1\}$ such that $(x_0,y_0) \in \Omega_l$. 
    Since $L \neq 0$, $\Dr \in \CAD_{\text{3-irr}}(\Fc)$, and since $\Psi_{\psi_A(B)}= \Psi_{B}$ is defined on $\Tree(\Dr, \Fc)$, there exists $S_i \in \Fc$ such that $D_{I'}\subseteq S_i$ and $S_i$ has no curtain at $\{(x_0,y_0)\}$ (since otherwise the CAD$(\Fc)$ reduction rule $\Phi_{I'}$ would be defined on $\Dr$, which is absurd since it is irreducible at the last level).
    Using Corollary \ref{cor:bypass}, there exists a unique $b \in [-\infty;\infty]$ such that 
    $$\lim_{\substack{(x,y) \to (x_0,y_0)\\(x,y)\in D_{(a_1-1, b_2-1)}}}\delta_{I'}(x,y) = b.$$
    Moreover, if $b \notin \{-\infty, \infty\}$ (otherwise $(x_0,y_0) \in \Omega_0 \cup \Omega_{L+1}$), then $(x_0,y_0,b) \in S_i$. We suppose that this point belongs to some sector $D_{(a_1-1,b_2,2l+1)}$ of $\Dr$, and we obtain a contradiction. Since $\Dr$ is adapted to $S_i$, the whole sector $D_{(a_1-1,b_2,2l+1)}$ is a subset of $S_i$. Thus, $S_i$ has a curtain at $D_{(a_1-1,b_2)}$, which is absurd by the choice of $S_i$. 
    We obtain that the point $(x_0,y_0,b)$ belongs to some section $D_{(a_1-1,b_2,2l)}$ of $\Dr$ for some $l \in \{1, \ldots, L\}$, i.e. $(x_0,y_0) \in \Omega_l$.
    
    The definition of the limit readily implies Assertion \ref{item:disjoint}. %(see Lemma 3.3.12 of \cite{arnon})
    
    Note that part of the following proof of Assertion \ref{item:open} is similar to the proof of Theorem 3.4.4 of \cite{arnon}. 
    Let $(x_0,y_0) \in \Omega_l$ for some $l \in \{0, \ldots, L+1\}$. Since $\delta_{(a_1-1,b_2,0)}<\delta_{(a_1-1,b_2,2)} < \ldots < \delta_{(a_1-1,b_2,2L)}< \delta_{(a_1-1,b_2,2(L+1))}$ on $D_{(a_1-1,b_2)}$, there exist an open ball $U \subset \R^2$ centred at $(x_0,y_0)$ and an open interval $I\subseteq [-\infty, \infty]$ containing $\delta_{(a_1-1,b_2,2l)}(x_0,y_0)$ such that $D_{(a_1-1,b_2,2l)}$ is the only section of $\Dr$ built above $D_{(a_1-1,b_2)}$ which meets $\overline{U\times I}$ (the closure of $U \times I$ in $\R^2 \times [-\infty, \infty]$). Moreover, using the fact that $(x_0,y_0) \in \Omega_l$, there exists an open neighbourhood $V$ of $(x_0,y_0)$ in $\R^2$ such that $\delta_{I'}(V \cap D_{(a_1-1,b_2-1)}) \subset I$. We show that the open neighbourhood $N = U \cap V \cap D_{(a_1-1,b_2)}$ of $(x_0,y_0)$ in $D_{(a_1-1,b_2)}$ is a subset of $\Omega_l$. Let $(x_1,y_1) \in N$. By Item \ref{item:OmegaCovering}, there exists $k \in \{0, \ldots, L+1\}$ such that $(x_1,y_1) \in \Omega_k$. Since $N \subseteq V$, we must have  $\delta_{(a_1-1,b_2,2k)}(x_1,y_1) \in \overline{I}$. Since $N \subseteq U$, and since $(x_1,y_1, \delta_{(a_1-1,b_2,2k)}(x_1,y_1)) \in \overline{U \times I}$, we must have $k = l$.
    
    The existence of the reduction rule $\Phi_B$ form $\Cr$ to $\Er$ implies that $C_{B} \subseteq \Omega_j$, and hence Assertion \ref{item:notempty} follows. 
%\end{proof}

\subsection{The CAD $\Fr$ is smaller than the CAD $\Er$}\label{subsection:FpreceqC}

We consider the CAD $\Fr$ as in Subsection \ref{case:|A|=|B|} or \ref{case:|A|neq|B|}, and we show that $\Fr \preceq \Er$. It is equivalent to prove that for every cell $E_I \in \Er$, there exists a cell $F_J \in \Fr$ such that $E_I \subseteq F_J$.
Using Lemma \ref{lemma:psiBij}, we know that 
    \begin{align*}
        E_{I} = 
        \begin{cases}
            C_{I} \cup C_{I+ e_{|B|}} \cup C_{I + 2 e_{|B|}} &\text{ if } I \in S_{B - e_{|B|}}, \\
            C_{I + 2 e_{|B|}} &\text{ if } I \in S_B \cup N_B, \\
            C_{I} &\text{ otherwise.}
        \end{cases}
    \end{align*}
    By construction of $\Fr$ (see Subsections \ref{case:|A|=|B|} and \ref{case:|A|neq|B|}), every cell $C_{I'} \in \Cr$ is contained in some cell of $\Fr$, and we can track its index with respect to $I'$. 
    More precisely, we have 
$$C_{I'} \subseteq D_{\psi_A(I')} \subseteq F_{\psi_{\psi_A(B)}(\psi_A(I'))}.$$ 
If $I \notin S_{B - e_{|B|}}$, then $E_I \in \Cr$ and is therefore a subset of some cell of $\Fr$.
If $I \in S_{B - e_{|B|}}$, then we need to show that the three cells $C_I,  C_{I+ e_{|B|}}$ and $C_{I + 2 e_{|B|}}$ constituting $E_I$ are all subsets of the same cells of $\Fr$.
Equivalently, we must verify the equalities
\begin{align*}
    \psi_{\psi_A(B)}(\psi_A(I)) = \psi_{\psi_A(B)}(\psi_A(I+e_{|B|})) = \psi_{\psi_A(B)}(\psi_A(I+2e_{|B|})).
\end{align*}
A straightforward computation shows that this formula holds in each relevant cases. 
This concludes the proof of Theorem \ref{thrm:main-curtained}.
\end{proof}

\section{Further Work}

First, the results of this paper are primarily theoretical. To move towards an effective implementation of an algorithm computing minimum CADs, it will be necessary to develop concrete and computationally feasible criteria for deciding when a tree reduction lifts to a CAD reduction. We conjecture that the adjacency relations between cells (see for instance \cite{ArnonCollinsMcCallum3D} and \cite{strzebonski}) will play a central role in formulating such criteria.

Next, the grail would be the design of a projection operator whose associated CAD algorithm directly produces a minimum CAD adapted to the input. Such an operator is unlikely to rely only on the classical discriminant and resultant based projection operators, since these inherently consider all complex roots rather than the real roots only.

Finally, the techniques employed here to prove Theorem \ref{thrm:main} seem to not generalize immediately in $\R^n$ with $n \geqslant 4$. First, the case distinction would be more intricate. Second, Corollary \ref{cor:fibreSpace} fails in general in $\R^4$ (see Proposition 3.5 of \cite{exotic}).
\begin{quote}
\textbf{Question.}  
Does every finite family of algebraic sets in $\mathbb{R}^n$ ($n \in \N^*$) admit a minimum CAD?
\end{quote}
This question remains open for $n \geqslant 4$.

\begin{acks}
The author would like to thank Pierre Mathonet for numerous helpful suggestions, Naïm Zénaïdi for fruitful discussions and the reviewers for their valuable feedback and comments. The author is supported by the FNRS-DFG PDR Weaves (SMT-ART) grant 40019202.
\end{acks}

%\nocite{*}
\bibliographystyle{ACM-Reference-Format}
%\bibliography{main}


\end{document}